\documentclass{sig-alternate-2013}

\usepackage[
  pass,
]{geometry}

\newfont{\mycrnotice}{ptmr8t at 7pt}
\newfont{\myconfname}{ptmri8t at 7pt}
\let\crnotice\mycrnotice%
\let\confname\myconfname%

\usepackage{amsthm,amssymb,amsmath,stmaryrd}
\usepackage{graphics}
\usepackage{bbm}
\usepackage{tikz}
\usepackage[plainpages=false,pdfpagelabels=false,colorlinks=true,citecolor=blue,hypertexnames=false]{hyperref}
\usepackage{color}

\let\mod\bmod

\newtheorem{theorem}{Theorem}[section]
\newtheorem{notation}[theorem]{Convention}
\newtheorem{prop}[theorem]{Proposition}

\newtheorem{lemma}[theorem]{Lemma}
\newtheorem{remark}[theorem]{Remark}

\newtheorem{definition}[theorem]{Definition}
\newtheorem{example}[theorem]{Example}

\newcommand{\bB}{ {\mathbb B}}
\newcommand{\bE}{ {\mathbb E}}
\newcommand{\bF}{ {\mathbb F}}
\newcommand{\bD}{ {\mathbb D}}
\newcommand{\bP}{ {\mathbb P}}
\newcommand{\bN}{ {\mathbb N}}
\newcommand{\bZ}{ {\mathbb Z}}
\newcommand{\bU}{ {\mathbb U}}
\newcommand{\bV}{ {\mathbb V}}
\newcommand{\bW}{ {\mathbb W}}

\newcommand{\cP}{ {\mathcal P}}

\newcommand{\cC}{ {\mathcal C}}
\newcommand{\spanning}{\operatorname{span}}

\newcommand{\lsx}{\bF\langle S_x \rangle}

\let\set\mathbbm
\def\lc{\operatorname{lc}}
\def\im{\operatorname{im}}

\permission{Permission to make digital or hard copies of all or part of this work for personal or classroom use is granted without fee provided that copies are not made or distributed for profit or commercial advantage and that copies bear this notice and the full citation on the first page. Copyrights for components of this work owned by others than ACM must be honored. Abstracting with credit is permitted. To copy otherwise, or republish, to post on servers or to redistribute to lists, requires prior specific permission and/or a fee. Request permissions from permissions@acm.org.}
\conferenceinfo{ISSAC'15,}{July 6--9, 2015, Bath, United Kingdom.\\
{\mycrnotice{Copyright is held by the owner/author(s). Publication rights licensed to ACM.}}}
 \copyrightetc{Copyright \copyright~2015 ACM \the\acmcopyr}
\crdata{978-1-4503-3435-8/15/07\ ...\$15.00.\\
DOI:  http://dx.doi.org/10.1145/2755996.2756648}

\begin{document}
\title{A Modified Abramov-Petkov\v{s}ek Reduction and \\ Creative Telescoping for
Hypergeometric Terms\titlenote{S.\ Chen was supported in part by the NSFC grant 11371143 and by
President Fund of Academy of Mathematics and Systems Science, CAS (2014-cjrwlzx-chshsh), H.\ Huang
by the Austrian Science Fund (FWF) grant W1214-13, M.\ Kauers by the Austrian Science Fund (FWF)
grants Y464-N18 and F50-04, H.\ Huang and Z.\ Li by two NSFC grants (91118001, 60821002/F02) and a
973 project (2011CB302401).}}

\numberofauthors{1} 
\author{\medskip
        Shaoshi Chen$^1$, \, Hui Huang$^{1,2}$, \, Manuel Kauers$^2$,  \, Ziming Li$^1$ \\
        \smallskip
        \affaddr{$^1$KLMM,\, AMSS, \,Chinese Academy of Sciences, Beijing 100190, (China)}\\
        \smallskip
        \affaddr{$^2$RISC, Johannes Kepler University, Linz A-4040, (Austria)}\\
        \smallskip
        \email{{schen, huanghui}@amss.ac.cn}\\
        \email{mkauers@risc.jku.at, \, zmli@mmrc.iss.ac.cn}
    }

\maketitle
\begin{abstract}
The Abramov-Petkov\v{s}ek reduction computes an additive decomposition of a hypergeometric term,
which extends the functionality of the Gosper algorithm for indefinite hypergeometric summation. We
modify the Abramov-Petkov\v{s}ek reduction so as to decompose a hypergeometric term as the sum of a
summable term and a non-summable one. The outputs of the Abramov-Petkov\v{s}ek reduction and our
modified version share the same required properties. The modified reduction does not solve  any
auxiliary linear difference equation explicitly. It is also more efficient than the original
reduction according to computational experiments. Based on this reduction, we design a new
algorithm to compute minimal telescopers for bivariate hypergeometric terms. The new algorithm can
avoid the costly computation of certificates.
\end{abstract}




\section{Introduction}\label{SECT:intro}

    Creative telescoping is a staple of symbolic summation. Its main use
    is to construct recurrence equations that have a prescribed definite sum among
    their solutions. By using other algorithms applicable to recurrence equations, it
    is then possible to derive interesting facts about the original definite sum, such
    as closed forms or asymptotic expansions.

    The computational problem of creative telescoping is to construct, for a given
    term $f(x,y)$, polynomials $\ell_0,\dots,\ell_r$ in~$x$ only, not all zero, and
    another term $g(x,y)$ s.t.\
    \[
    \ell_0(x)f(x,y)+\cdots+\ell_r(x)f(x+r,y)=g(x,y+1)-g(x,y).
    \]
        The number $r$ may or may not be part of the input.

    We can distinguish four generations of creative telescoping algorithms. The
    first generation was based on elimination
    techniques~\cite{Fasenmyer1947,Zeilberger1990,PWZBook1996,chyzak98}. The second generation
    started with what is now known as Zeilberger's
    algorithm~\cite{zeilberger90a,Almkvist1990,zeilberger91,PWZBook1996}. The
    algorithms of this generation use the idea of augmenting an algorithm for
    indefinite summation (or integration) by additional parameters
    $\ell_0,\dots,\ell_r$ that are carried along during the calculation and are
    finally instantiated, if at all possible, such as to ensure the existence of a
    term~$g$ as needed for the right-hand side.
    See \cite{PWZBook1996} for details about the first two generations.

    The third generation was initiated by Apagodu and
    Zeilberger~\cite{mohammed05,apagodu05}. In a sense, they applied a
    second-generation algorithm by hand to a generic input and worked out the
    resulting linear system of equations for the parameters~$\ell_i$ and the
    coefficients inside the desired term~$g$. Their algorithm then merely consists
    in solving this system. This approach is interesting not only because it is
    easier to implement and tends to run faster than earlier algorithms, but also
    because it is easy to analyze. In fact the analysis of algorithms from this
    family gives rise to the best output size estimates for creative telescoping
    known so far~\cite{chen12c,chen12b,chen14a}. A disadvantage is that these algorithms
    may not always find the smallest possible output.

    The fourth generation of creative telescoping algorithms originates
    from~\cite{BCCL2010}. The basic idea behind these algorithms is to bring each term
    $f(x+i,y)$ of the left-hand side into some kind of normal forms modulo all terms
    that are differences of other terms. Then to find $\ell_0,\dots,\ell_r$ amounts
    to finding a linear dependence among these normal forms. The key advantage
    of this approach is that it separates the computation of the $\ell_i$ from the
    computation of~$g$. This is desirable in the typical situation where we are only
    interested in the $\ell_i$ and their size is much smaller than the size
    of~$g$. With previous algorithms there was no way to obtain~$\ell_i$ without
    also computing~$g$, but with fourth generation algorithms there is.
    So far this approach has only been worked out for several instances in
    the differential case~\cite{BCCL2010,bostan13,BCCLX2013}. The goal of the present paper is to give
    a fourth-generation algorithm for the discrete case, namely for the classical
    setting of hypergeometric telescoping.

    Our starting point is the Abramov-Petkov\v sek reduction for hypergeometric
    terms introduced in~\cite{Abramov2001} and summarized in
    Section~\ref{SECT:apreduction} below.  Unfortunately the reduced forms obtained
    by this reduction are not sufficiently \lq\lq normal\rq\rq~for our purpose. Therefore, in
    Sections~\ref{SECT:iapreduction} and~\ref{SECT:sum} we present a refined variant
    of the reduction process and show that the corresponding normal forms are
    well-behaved with respect to taking linear combinations. Then in
    Section~\ref{SECT:tele} we describe the creative telescoping algorithm obtained
    from this reduction. The final section contains an experimental comparison
    between this algorithm and the built-in algorithm of Maple.

    \section{Preliminaries}\label{SECT:hts}
    Throughout the paper, we let~$\bF$ be a field of characteristic zero, and~$\bF(y)$ be the field of rational functions
    in~$y$ over~$\bF$. Let~$\sigma_y$ be the automorphism that maps~$r(y)$ to~$r(y+1)$ for every~$r \in \bF(y)$.
    The pair~$(\bF(y), \sigma_y)$ is called a difference field. A difference ring extension of~$(\bF(y),\sigma_y)$
        is a ring $\bD$ containing~$\bF(y)$ together with a distinguished endomorphism $\sigma_y\colon\bD\to\bD$
        whose restriction to $\bF(y)$ agrees with the automorphism defined before.
        An element $c\in\bD$ is called a constant if $\sigma_y(c)=c$.
        For a nonzero polynomial~$p \in \bF[y]$, its degree and leading
    coefficient are denoted by~$\deg_y(p)$ and~$\lc_y(p)$, respectively.
    \begin{definition}\label{DEF:ht}
        Let~$\bD$ be a difference ring extension of~$\bF(y)$. A nonzero element~$T \in \bD$ is called
        a {\em hypergeometric term} over~$\bF(y)$ if $\sigma(T) = r T$ for some~$r \in \bF(y)$. We call~$r$ the
        {\em shift quotient} of~$T$ w.r.t.~$y$.
    \end{definition}
    The product of hypergeometric terms is again hypergeometric. Two hypergeometric terms~$T_1, T_2$ are
    called {\em similar} if there exists a rational function~$r\in \bF(y)$ s.t.~$T_1 = r T_2$. By Proposition~5.6.2
    in~\cite{PWZBook1996}, the sum of similar hypergeometric terms is either hypergeometric or zero.

    A univariate hypergeometric term~$T$ is called {\em hypergeometric summable}
    if there exists another hypergeometric term~$G$ s.t.~$T = \Delta_y(G)$, where~$\Delta_y$ denotes the difference
    of~$\sigma_y$ and the identity map. We abbreviate \lq\lq hypergeometric summable\rq\rq\ as \lq\lq summable\rq\rq\ in this paper.

    Given a hypergeometric term~$T$, we let~$\bU_T$ be the union of~$\{0\}$ and the set of
    summable hypergeometric terms that are similar to~$T$. Then~$\bU_T$ is an $\bF$-linear subspace
    of~$\bD$. Note that~$\bU_T=\bU_H$ if~$H$ is a hypergeometric term similar to~$T$.

    Recall \cite[\S 1]{Abramov2001} that a nonzero polynomial in~$\bF[y]$ is said to be {\em shift-free}
    if it is coprime with any of its nontrivial shifts.
    A nonzero rational function is said to be {\em shift-reduced} if its numerator is coprime with any
    shift of its denominator.

    A basic property of shift-reduced rational functions is given below.
    \begin{lemma} \label{LM:shiftreduced}
        Let~$f \in \bF(y)$ be shift-reduced and unequal to one. If there exists~$r \in \bF[y]$
        s.t.~$f \sigma_y(r) -  r = 0,$
        then~$r = 0$. 
    \end{lemma}
    \begin{proof}
        Suppose that~$r \neq 0$. Then~$f = r/\sigma_y(r)$. Since~$f$ is unequal to one, $r$ does not belong to~$\bF$.
        It follows that~$f$ is not shift-reduced, a contradiction.
    \end{proof}

    According to~\cite{Abramov2001, AbramovPetkovsek2002b}, every hypergeometric term~$T$ has a multiplicative
    decomposition~$S H$, where~$S$ is in~$\bF(y)$ and~$H$ is another hypergeometric term whose shift quotient is shift-reduced.
    We call the shift quotient~$K :=\sigma_y (H) /H$ a {\em kernel} of~$T$ w.r.t.~$y$ and~$S$ the corresponding {\em shell}.
    Note that~$K=1$ if and only if~$T$ is a rational function, which is then equal to~$cS$ for some element $c\in \bD$
        with $\sigma_y(c)=c$.

    Let~$T=SH$ be a multiplicative decomposition, where~$S$ is a rational function
    and~$H$ a hypergeometric term with a kernel~$K$.
    Assume that~$T=\Delta_y(G)$ for some hypergeometric term~$G$. A straightforward calculation shows that~$G$
    is similar to~$T$. So there exists~$r \in \bF(y)$ s.t.~$G = r H$. One can easily verify that
    \begin{equation}\label{EQ:summable}
        SH  = \Delta_y(rH)\, \Longleftrightarrow \, S = K \sigma_y(r)-r.
    \end{equation}
    Let~$\bV_K = \{ K \sigma_y(r)-r \mid r \in \bF(y) \}$, which is an $\bF$-linear subspace of~$\bF(y)$.
    Then~\eqref{EQ:summable} translates into
    \begin{equation} \label{EQ:congsum}
        S H \equiv 0 \mod \bU_H \, \Longleftrightarrow \, S \equiv 0 \mod \bV_K.
    \end{equation}
    These congruences enable us to shorten expressions.

    \section{Abramov-Petkov\v{s}ek reduction}\label{SECT:apreduction}
    Reduction algorithms have been developed for computing additive decompositions of rational
    functions~\cite{Abramov1975}, hyperexponential functions~\cite{BCCLX2013}, and
    hypergeometric terms~\cite{Abramov2001, AbramovPetkovsek2002b}. These algorithms can be viewed as generalizations of
    the Gosper  algorithm~\cite{Gosper1978,PWZBook1996} and its differential analogue~\cite{Almkvist1990}.

    The Abramov-Petkov\v{s}ek reduction~\cite{Abramov2001, AbramovPetkovsek2002b} is fundamental for this paper.
    To describe it concisely, we need a notational convention and a technical definition.
    \begin{notation} \label{CON:kernel}
        Let~$T$ be a hypergeometric term whose kernel is~$K$ and the corresponding shell is~$S$. Then~$T=SH$, where~$H$ is a hypergeometric term whose shift quotient is~$K$. Assume that~$K$ is unequal to one.
        Moreover, write~$K = u/v$, where~$u,v$ are polynomials in~$\bF[y]$ with~$\gcd(u,v)=1$.
    \end{notation}
    \begin{definition}\label{DEF:prime}
        A nonzero polynomial~$p$ in~$\bF[y]$ is said to be {\em strongly prime} with~$K$
        if~$\gcd\left( p, \sigma_y^{-i}(u) \right) {=} \gcd\left( p, \sigma_y^{i}(v) \right) {=} 1$
        for all~$i \ge 0$.
    \end{definition}

    The proof of Lemma~3 in~\cite{Abramov2001} contains a reduction algorithm whose inputs and outputs
    are  given below.

    \medskip \noindent
    {\bf AbramovPetkov\v{s}ekReduction:}
    Given~$K$ and~$S$ as defined in Convention~\ref{CON:kernel},  compute
    a rational function~$S_1 {\in} \bF(y)$ and polynomials~$b, w \in \bF[y]$ s.t.~$b$ is shift-free and
    strongly prime with~$K$, and the following equation holds:
    \begin{equation} \label{EQ:ap}
        S =  K \sigma_y(S_1) - S_1  + \frac{w}{ b\cdot  \sigma_y^{-1}(u)\cdot v}.
    \end{equation}


    The algorithm contained in the proof of Lemma~3
    in~\cite{Abramov2001} is described as pseudo code on page~4 of the same paper,
    in which the last ten lines are to make the denominator of the rational function~$V$ in its output
    minimal in some technical sense. We shall not execute these lines. Then the {algorithm} will compute
    two rational functions~$U_1$ and~$U_2$. They correspond to~$S_1$ and~$w/\left(
    b\sigma_y^{-1}(u)v\right)$ in~\eqref{EQ:ap}, respectively.

    We slightly modify the output of the Abramov-Petkov\v{s}ek reduction.
    Note that~$K$ is shift-reduced and~$b$ is strongly prime with~$K$.
    Thus,~$b$, $\sigma_y^{-1}\left( u\right)$ and~$v$ are pairwise coprime.
    By partial fraction decomposition,~\eqref{EQ:ap} can be rewritten as
    \[   S = K \sigma_y(S_1) - S_1  +  \left(\frac{a}{b} + \frac{p_1}{\sigma_y^{-1}(u)}+\frac{p_2}{v}\right),  \]
    where~$a, p_1, p_2 \in \bF[y]$.
    Furthermore, we set~$r = p_1/\sigma_y^{-1}(u)$. A direct calculation yields
    $r  = K \sigma_y(-r) - (-r) + \sigma_y(p_1)/v$.
    Update~$S_1$ to be~$S_1 - r$ and set~$p$ to be~$\sigma_y(p_1) + p_2$. Then
    \begin{equation} \label{EQ:sap}
        S = K \sigma_y(S_1) - S_1 +  \frac{a}{b} + \frac{p}{v}.
    \end{equation}
    This modification leads to shell reduction specified below.

    \medskip \noindent
    {\bf ShellReduction}:
    Given~$K$ and~$S$ as defined in Convention~\ref{CON:kernel},  compute
    a rational function~$S_1 \in \bF(y)$ and polynomials~$a, b, p \in \bF[y]$ s.t.~$b$ is shift-free and
    strongly prime with~$K$, and that~\eqref{EQ:sap} holds.

    Shell reduction provides us with a necessary condition on summability.
    \begin{prop} \label{PROP:shell}
        With Convention~\ref{CON:kernel}, assume that~$a, b, p$ are polynomials in~$\bF[y]$  s.t.~$b$ is shift-free
        and strongly prime with~$K$. Assume further that~\eqref{EQ:sap} holds. If~$T$ is summable, then~$a/b$ belongs to~{$\bF[y]$}.
    \end{prop}
    \begin{proof} Recall that~$T=SH$ by Convention~\ref{CON:kernel} and it has a kernel~$K$ and the corresponding shell~$S$.
        It follows from \eqref{EQ:congsum} and~\eqref{EQ:sap} that
        $T \equiv \left( a/b+p/v\right) H \mod \bU_H$. Thus,~$T$ is summable if and only if~$\left( a/b+p/v\right) H$
        is summable.

        Set~$H^\prime=(1/v)H$, which has a kernel~$K^\prime=u/\sigma_y(v)$. 
        Note that~$b$ is also strongly prime with~$K^\prime$. We can apply  Theorem~11 in~\cite{AbramovPetkovsek2002b}
        to~$\left(av/b+p\right)H^\prime$, which equals~$\left( a/b+p/v\right) H.$
        Thus,~$a/b$ is a polynomial because~$b$ is coprime with~$v$.
    \end{proof}
    
    \begin{example}\label{EX:nonsummable}
        Let~$T = y^2 y!/(y+1)$. Then the term has a kernel~$K = y+1$ and the shell~$S = y^2/(y+1)$. Shell reduction yields
        $ S \equiv {-1/(y+2) + y/v }\mod \bV_K$
        where~$v {=} 1$.
        By Proposition~\ref{PROP:shell},~$T$ is not summable.
    \end{example}
    \noindent Note that~$a{/}b+p{/}v$ in \eqref{EQ:sap} can be
    nonzero for a summable~$T$.
    \begin{example}\label{EX:summable}
        Let~$T = y \cdot y!$ whose kernel is~$K =
        y+1$ and shell is~$S = y$.
        Then~$S \equiv  y/v \mod \bV_K$, where~$v=1$. But~$T$ is summable as it is equal
        to~$\Delta_y\left(y!\right)$.
    \end{example}
    The above example illustrates that neither shell reduction nor the Abramov-Petkov\v{s}ek reduction can decide
    summability directly. One way to proceed is to
    find a polynomial solution of an
    auxiliary first-order linear difference equation~\cite{AbramovPetkovsek2002b}.
    We show how this can be avoided in the next section.

    \section{Modifications} \label{SECT:iapreduction}
    After the shell reduction described in~\eqref{EQ:sap}, it
    remains to check the summability of~$\left(a/b+p/v\right)H$. In the
    rational case, i.e.\ when the kernel~$K$ is one,~$a{/}b + p{/}v$
    in~\eqref{EQ:sap} can be further reduced to~$a{/}b$ with~$\deg_y(a) < \deg_y(b)$, because all polynomials are
    rational summable. However, a hypergeometric term with a polynomial
    shell is not necessarily summable, for example, $y!$ has a polynomial shell but it is not summable.

    We define the notion of discrete residual forms for rational functions, and  present
    a discrete variant of the polynomial reduction for
    hyperexponential functions given in~\cite{BCCLX2013}. This variant
    not only leads to a direct way to decide summability, but also
    reduces the number of terms of~$p$ in~\eqref{EQ:sap}.

    \subsection{Discrete residual forms} \label{SUBSECT:residual}
    With Convention \ref{CON:kernel}, we define an~$\bF$-linear
    map~$\phi_K$ from~$\bF[y]$ to itself by sending~$p$ to~$u
    \sigma_y(p) - v p$ for all~$p \in \bF[y]$. We call~$\phi_K$ the
    {\em map for polynomial reduction w.r.t.~$K$}.
    \begin{lemma} \label{LM:direct}
        Let
        \[ \bW_K = \spanning_{\bF}\left\{y^{\ell} \mid \ell \in {\mathbb{N}}, \ell
        \neq \deg_y (p) \text{ for all } p \in \im\left(\phi_K \right) \right\}.\]
        Then~$\bF[y] = \im\left( \phi_K \right) \oplus \bW_K$.
    \end{lemma}
    \begin{proof}
        By the definition of~$\bW_K$, $\im\left(\phi_K\right) \cap \bW_K = \{0\}$.
        The same definition also implies that, for every non-negative integer~$m$, there exists a polynomial~$f_m \in \im\left(\phi_K\right) \cup \bW_K$
        s.t.~the degree of~$f_m$ is equal to~$m$.  The set~$\{f_0, f_1, f_2, \ldots \}$ forms an $\bF$-basis of~$\bF[y]$. Thus $\bF[y] = \im\left( \phi_K \right) \oplus \bW_K$.
    \end{proof}
    In view of the above lemma, we
    call~$\bW_K$ the {\em standard complement of~$\im(\phi_K)$}.
    A polynomial~$p \in \bF$ can be uniquely decomposed as~$p=p_1+p_2$
    with~$p_1 \in \im\left( \phi_K\right)$ and~$p_2 \in \bW_K$.
    \begin{lemma} \label{LM:reducedform}
        With Convention~\ref{CON:kernel}, let~$p$ be a polynomial in~$\bF[y]$.
        Then there exists~$q \in \bW_K$ s.t.~$p/v \equiv q/v \mod \bV_K$.
    \end{lemma}
    \begin{proof}
        Let~$q$ be the projection of~$p$ on~$\bW_K$. Then there exists~$f$ in~$\bF[y]$ s.t.~$p = \phi_K(f) + q$,
        that is,~$p = u \sigma_y(f) - vf + q$.
        So~$p/v = K \sigma_y(f)-f+q/v$, that is,~$p/v \equiv q/v \mod \bV_K$.
    \end{proof}
    \begin{remark} \label{RE:poly}
        Replacing~$p$ in the above lemma by~$vp$, we see that, for every polynomial~$p \in \bF[y]$,
        there exists~$q \in \bW_K$ s.t.~$p \equiv q/v \mod \bV_K$.
    \end{remark}
    By Lemma~\ref{LM:reducedform} and Remark~\ref{RE:poly}, \eqref{EQ:sap} implies that
    \begin{equation} \label{EQ:iap}
        S \equiv \frac{a}{b} + \frac{q}{v} \mod \bV_K,
    \end{equation}
    where~$a, b, q \in \bF[y]$, $\deg_y(a)<\deg_y(b)$, $b$ is shift-free and strongly prime with~$K$,
    and~$q \in \bW_K$.
    The congruence~\eqref{EQ:iap} motivates us to translate the notion of (continuous) residual forms in~\cite{BCCLX2013} into the
    discrete setting.
    \begin{definition} \label{DEF:residual}
        With Convention~\ref{CON:kernel}, we further let~$f$ be a rational function in~$\bF(y)$.
        Another rational function~$r$ in~$\bF(y)$ is called a {\em (discrete) residual form} of~$f$ w.r.t.~$K$
        if there exist~$a,b,q$ in~$\bF[y]$ s.t.
        $$f \equiv r \mod \bV_K \quad \text{and} \quad r = \frac{a}{b}+\frac{q}{v},$$
        where~$\deg_y(a)<\deg_y(b)$, $b$ is shift-free and strongly
        prime with~$K$, and~$q$ belongs to~$\bW_K$. For brevity, we just say that~$r$ is {\em a residual form}
        w.r.t.~$K$ if~$f$ is clear from the context.
    \end{definition}

    Residual forms help us decide summability, as shown in the next proposition.
    \begin{prop}\label{PROP:residual}
        With Convention~\ref{CON:kernel}, we further assume that~$r$ is a
        nonzero residual form w.r.t.~$K$. The hypergeometric
        term~$rH$ is not summable.
    \end{prop}
    \begin{proof}
        Suppose that~$rH$ is summable. Let~$r=a/b+q/v$,
        where~$\deg_y(a)<\deg_y(b)$, $b$ is shift-free and strongly
        prime with~$K$, and~$q$ belongs to~$\bW_K$.
        By Proposition~\ref{PROP:shell},~$a/b$ is a polynomial.
        Since~$\deg_y(a)<\deg_y(b)$, $a=0$. Thus,~$(q/v)H$ is summable.
        It follows from~\eqref{EQ:summable} that there exists~$w$ in~$\bF(y)$
        s.t.~$u\sigma_y(w)-vw=q$. Thus,~$w \in \bF[y]$ by Theorem~5.2.1 in~\cite[page 76]{PWZBook1996}.
        So~$q$ belongs to~$\im\left(\phi_K\right)$. But~$q$ also belongs to~$\bW_K$.
        By Lemma~\ref{LM:direct},~$q=0$, a contradiction.
    \end{proof}

    With Convention~\ref{CON:kernel}, let~$r$ be a residual form of the shell~$S$.
    Then~$SH \equiv rH \mod \bU_H$ by~\eqref{EQ:congsum} and~\eqref{EQ:iap}.
    By Proposition~\ref{PROP:residual}, $SH$ is summable if and only if~$r=0$.
    Thus, determining the summability of a hypergeometric term~$T$ amounts to computing a residual
    form of the corresponding shell w.r.t.~a kernel of~$T$, which is studied below.
    \subsection{Polynomial reduction}
    To compute a residual form of a rational function, we project a polynomial on
    $\im(\phi_K)$ and on its standard complement~$\bW_K$, both defined in the previous subsection.

    Let~$\bB_K=\left\{ \phi_K(y^i) \mid i \in \bN \right\}$.
    The $\bF$-linear map~$\phi_K$ is injective by Lemma~\ref{LM:shiftreduced}. So~$\bB_K$
    is an $\bF$-basis of~$\im\left( \phi_K \right)$, which allows us to construct an echelon basis.
    By an echelon basis, we mean an $\bF$-basis in which distinct elements have distinct degrees.
    We can easily project a polynomial using an echelon basis and linear elimination.

    To construct an echelon basis, we rewrite~$\im(\phi_K)$ as
    \begin{equation*}
        \im(\phi_K) = \left\{u\Delta_y(p) - (v - u)p \mid p \in \bF[y]
        \right\}.
    \end{equation*}
    Set~$\alpha_1 = \deg_y(u), \alpha_2 = \deg_y(v)$, and~$\beta = \deg_y(v -
    u)$. Moreover, set~$\tau_K = \lc_y(v - u){/}\lc_y(u)$, which is nonzero due to Convention~\ref{CON:kernel} and let~$p$ be a nonzero polynomial in~$\bF[y]$.

    We make the following case distinction.

    \smallskip \noindent
    {\em Case 1.}~$\beta > \alpha_1$. Then~$\beta = \alpha_2$, and
    \begin{equation*}
        \phi_{K}(p) = - \lc_y(v - u)\lc_y(p)y^{\alpha_2 + \deg_y(p)} + \text{ lower terms}.
    \end{equation*}
    So~$\bB_K$ is an echelon basis of~$\im(\phi_{K})$, in
    which~$\deg_y(\phi_{K}(y^i))$ is equal to~$\alpha_2 + i$ for all~$i \in \bN$.
    Accordingly,~$\bW_{K}$ has an echelon basis~$\left\{1, y, \ldots,
    y^{\alpha_2 - 1}\right\}$ and~$\dim(\bW_K)=\alpha_2$.

    \smallskip \noindent
    {\em Case 2.}~$\beta = \alpha_1$. Then
    \begin{equation*}
        \phi_{K}(p) = - \lc_y(v - u)\lc_y(p)y^{\alpha_1 + \deg_y(p)} + \text{ lower terms}.
    \end{equation*}
    So~$\bB_K$ is an echelon basis of~$\im(\phi_{K})$, in
    which~$\deg_y(\phi_{K}(y^i))$ is equal to~$\alpha_1 + i$ for all~$i \in \bN$.
    Accordingly,~$\bW_{K}$ has an echelon basis~$\left\{1, y, \ldots,
    y^{\alpha_1 {-} 1}\right\}$ and~$\dim(\bW_K)=\alpha_1$.

    \smallskip \noindent
    {\em Case 3.}~$\beta < \alpha_1 - 1$. If~$\deg_y(p) = 0$, then~$\phi_{K}(p) = (u - v)p$. Otherwise, we have
    \begin{equation*}
        \phi_{K}(p) = \deg_y(p)\lc_y(u)\lc_y(p)y^{\alpha_1 + \deg_y(p) - 1} + \text{ lower terms}.
    \end{equation*}
    It follows that~$\bB_K$ is an echelon basis of~$\im(\phi_{K})$, in which~$\deg_y(\phi_{K}(1)) = \beta$ and
    \begin{equation*}
        \deg_y(\phi_{K}(y^i)) = \alpha_1 + i - 1 \quad \text{ for all } i \geq 1.
    \end{equation*}
    So~$\bW_{K}$ has an echelon basis~$\left\{1, \ldots, y^{\beta - 1}, y^{\beta + 1}, \ldots,
    y^{\alpha_1 - 1}\right\}$, and~$\dim(\bW_K)=\alpha_1-1$.

    \smallskip \noindent
    {\em Case 4.}~$\beta = \alpha_1 - 1$ and~$\tau_K$ is not a positive integer. Then
    \begin{align}
        \phi_{K}(p) &= \left(\deg_y(p)\lc_y(u) - \lc_y(v - u)\right)\lc_y(p)y^{\alpha_1
            + \deg_y(p) - 1} \nonumber\\
        &\quad + \text{ lower terms}.\label{EQ:case4}
    \end{align}
    So~$\bB_K$ is an echelon basis of~$\im(\phi_{K})$, in which, for all~$i \in \mathbb{N}$,
    $\deg_y(\phi_K(y^i)) = \alpha_1 + i - 1$.
    Accordingly,~$\bW_{K}$ is spanned by an
    echelon basis~$\left\{1, y, \ldots, y^{\alpha_1 - 2}\right\}$, and has dimension~$\alpha_1{-}1$.

    \smallskip \noindent
    {\em Case 5.}~$\beta = \alpha_1 - 1$ and~$\tau_K$ is a positive
    integer. It follows from~\eqref{EQ:case4} that for~$i \neq
    \tau_{K}$,~$\deg_y(\phi_{K}(y^i)) = \alpha_1 + i - 1$. Moreover, for
    every polynomial~$p$ of degree~$\tau_{K}$,~$\phi_{K}(p)$ is of
    degree less than~$\alpha_1 + \tau_{K} - 1$. So any echelon basis
    of~$\im(\phi_{K})$ does not contain a polynomial of degree~$\alpha_1
    + \tau_{K} - 1$. Set
    \begin{equation*}
        \bB_K^\prime = \left\{\phi_{K}(y^i) \mid i \in \bN, i \neq \tau_{K}\right\}.
    \end{equation*}
    Reducing~$\phi_{K}(y^{\tau_{K}})$ by the polynomials
    in~$\bB_K'$, we obtain a polynomial~$p'$ with~$\deg_y(p') <
    \alpha_1 - 1$. Since~$\bB_K$ is an~$\bF$-basis and~$\bB_K'
    \subset \bB_K$,~$p' \neq 0$. So~$\bB'_K \cup \{p'\}$ is
    an echelon basis of~$\im(\phi_{K})$. Consequently,~$\bW_{K}$ is
    spanned by an echelon basis~$\left\{1, y, \ldots, y^{\deg_y(p') {-}
        1}, y^{\deg_y(p') {+} 1}, \ldots, y^{\alpha_1 {-} 2}, y^{\alpha_1 {+}
        \tau_{K} - 1}\right\}$.
    The dimension of~$\bW_K$ is equal to~$\alpha_1-1$.

    \begin{example}\label{EX:case5}
        Let~$K = (y^4 + 1)/(y+1)^4$, which is shift-reduced. Then~$\tau_{K} = 4$. According to Case~5,~$ \im(\phi_{K})$ has an echelon basis
        \[ \left\{ \phi_K\left( p   \right) \right\}\cup
        \left\{ \phi_K\left(y^m \right) \mid m \in \bN, m \neq 4\right\},
        \]
        where~$p=y^4 + y/3 + 1/2$, $\phi_K(p) =  (5/3)y^2+2y+4/3$, and~$\phi_K\left(y^m\right) = (m-4)y^{m+3} + \text{lower terms}.$
        Therefore,~$\bW_{K}$ has a basis~$\{1,y,y^7\}$.
    \end{example}
    From the above case distinction and example, one observes that,
    although the degree of a polynomial in the standard complement
    depends on~$\tau_K$, which may be arbitrarily high, the number of its terms
    depends merely on the degrees of~$u$ and~$v$. We record this observation
    in the next proposition.
    \begin{prop}\label{PROP:termbound}
        With the Convention \ref{CON:kernel}, we further let
        $$\alpha_1 = \deg_y(u), \quad \alpha_2 = \deg_y(v), \quad \text{and} \quad \beta = \deg_y(v -
        u).$$ Then there exists~$\cP \subset \{y^i \mid i \in \mathbb{N}\}$ with
        $$|\cP| \leq \max \{\alpha_1, \alpha_2\}- \llbracket \beta \leq \alpha_1 -1\rrbracket$$
        s.t. every polynomial in~$\bF[y]$ can be reduced modulo~$\im(\phi_{K})$ to an~$\bF$-linear combination of the elements in~$\cP$. Note that here~$\llbracket \beta \leq \alpha_1 -1 \rrbracket$ equals~$1$ if~$\beta \leq \alpha_1 -1$, otherwise it is~$0$.
    \end{prop}
    \begin{proof}
        By the above case distinction, $\dim\left(\bW_{K}\right)$ is
        no more than~$\max \{\alpha_1, \alpha_2\}-\llbracket \beta \leq \alpha_1 -1\rrbracket$. The lemma follows.
    \end{proof}

    The above case distinction enables one to find an infinite sequence~$p_0, p_1, \ldots$ in~$\bF[y]$
    s.t.
    \begin{equation*} 
        \bE_K = \left\{ \phi_K(p_i) | i \in \bN \right\}~\text{with $\deg_y \phi_K(p_i) < \deg_y \phi_K (p_{i+1}),$}
    \end{equation*}
    is an echelon basis of~$\im\left(\phi_K\right)$.
    This basis allows us to project a polynomial on~$\im\left( \phi_K\right)$ and~$\bW_K$, respectively.
    In the first four cases, the~$p_i$'s can be chosen as powers of~$y$. But in the last case, one of the~$p_i$'s
    is not necessarily a monomial as shown in Example~\ref{EX:case5}.

    \medskip \noindent
    {\bf PolynomialReduction}:~Given~$p \in \bF[y]$,
    compute~$f \in \bF[y]$ and~$q \in \bW_K$ s.t.~$p=\phi_K(f)+q.$
    \begin{enumerate}
        \item If~$p=0$, then set~$f=0$ and~$q=0$; return.
        \item Set~$d=\deg_y(p)$. Find the subset~$\bP=\left\{p_{i_1}, \ldots, p_{i_s} \right\}$ consisting of the preimages
        of all polynomials
        in the echelon basis~$\bE_K$  whose degrees are at most~$d$.
        \item For~$k=s, s-1, \ldots, 1,$  perform linear elimination to find~$c_s,$ $c_{s-1}, \ldots,c_1\in\bF$
        s.t.~$p - \sum_{k=1}^s c_k \phi_K(p_{i_k}) \in \bW_K.$
        \item Set~$f=\sum_{k=1}^s c_k p_{i_k}$ and~$q=p- \phi_K(f)$; and return.
    \end{enumerate}
    We now present a modified version of the Abramov-Pet\-kov\-\v{s}ek reduction, which determines
    summability without solving any auxiliary difference equations explicitly.

    \medskip \noindent
    {\bf ModifiedAbramovPetkov\v{s}ekReduction}:~Given an irrational hypergeometric term~$T$ over~$\bF(y)$,
    compute a hypergeometric term~$H$ with a kernel~$K$,
    and two rational functions~$f, r \in \bF(y)$ s.t.~$r$ is a residual form w.r.t.~$K$, and
    \begin{equation} \label{EQ:iapred}
        T = \Delta_y(f H) + r H.
    \end{equation}
    \begin{enumerate}
        \item Find a kernel~$K$ and the corresponding shell~$S$ of~$T$;

                  \kern-\medskipamount
        \item Apply shell reduction to~$S$ w.r.t.~$K$ to find~$b, s, t \in \bF[y]$ and~$g \in \bF(y)$ s.t.~$b$ is shift-free
        and strongly prime with~$K$; and
        \begin{equation} \label{EQ:shellred}
            T = \Delta_y\left( g H\right) + \left( \frac{s}{b} + \frac{t}{v} \right) H,
        \end{equation}
        where~$\sigma_y(H)/H=K$ and~$v$ is the denominator of~$K$.

                  \kern-\medskipamount
        \item Set~$p$ and~$a$ to be the quotient and remainder of~$s$ and~$b$, respectively.

                  \kern-\medskipamount
        \item Apply polynomial reduction to~$vp+t$ to find~$h \in \bF[y]$ and~$q \in \bW_K$ s.t.~$vp+t = \phi_K(h) + q$.

                  \kern-\medskipamount
        \item Set~$f:=g{+}h$ and~$r:=a/b{+}q/v$ and return~$H$, $f$ and~$r$.
    \end{enumerate}
    \begin{theorem}\label{THM:iapred}
        With Convention~\ref{CON:kernel}, the modified version of the Abramov-Petkov\v{s}ek reduction computes
        a rational function~$f$ in~$\bF(y)$ and a residual form~$r$ w.r.t.~$K$ s.t.~\eqref{EQ:iapred} holds.
        Moreover,~$T$ is summable if and only if $r = 0$.
    \end{theorem}
    \begin{proof}
        Recall that~$T=SH$, where~$H$ has a kernel~$K$ and~$S$ is a rational function.
        Applying shell reduction to~$S$ w.r.t.~$K$ yields~\eqref{EQ:shellred}, which can be rewritten as
        \[  T = \Delta_y\left(g H\right) + \left( \frac{a}{b} + \frac{vp+t}{v} \right) H, \]
        where~$a$ and~$p$ are given in step~3 of the modified Abramov-Petkov\v{s}ek reduction. The polynomial
        reduction in step~4 yields that~$vp+t=u\sigma_y(h)-vh + q$. Substituting this into~\eqref{EQ:shellred},
        we see that
        \begin{align*}
            T & =  \Delta_y(gH) + \left(K \sigma_y(h)-h \right) H  + \left(\frac{a}{b} + \frac{q}{v}\right)H  \\
            & =  \Delta_y((g+h)H) +   rH,
        \end{align*}
        where~$r = a/b + q/v$. Thus,~\eqref{EQ:iapred} holds.
        By Proposition~\ref{PROP:residual},~$T$ is summable if and only~$r$ is equal to zero.
    \end{proof}
    \begin{example}\label{EX:nonsummable1}
        Let~$T$ be the same hypergeometric term as in
        Example~\ref{EX:nonsummable}. Then~$K = y+1$ and~$S = y^2/(y+1)$. Set~$H = y!$. By the shell reduction in Example~\ref{EX:nonsummable},
        \begin{equation*}
            T = \Delta_y\left(\frac{-1}{y+1} H\right) + \left(\frac{-1}{y+2} + \frac{y}{v}\right) H,
        \end{equation*}
        where~$v = 1$. Applying the polynomial reduction to~$(y/v)H$ yields ${ (y/v)H = \Delta_y(1\cdot H)}$.
        Combining the above steps, we decompose~$T$ as
        $
        T = \Delta_y\left(y/(y+1)H\right) - \left(1/(y+2)\right) H.
        $
        So the input term~$T$ is not summable.
    \end{example}
    \begin{example}\label{ex:summable1}
        Let~$T$ be the same hypergeometric term as in
        Example~\ref{EX:summable}. Then~$K = y+1$ and~$S = y$. Set~$H = y!$. By the shell reduction in Example~\ref{EX:summable},~$T = yH$. The polynomial reduction yields~$yH = \Delta_y\left(y!\right),$ hence~$T = \Delta_y\left(y!\right)$.
    \end{example}

\begin{remark}
With the notation given in the step~5 of the modified version, we can rewrite~$rH$
as~$\left(s_1/s_2 \right) G$, where~$s_1=av+bq$, $s_2=b$, and~$G=H/v$. It follows from the case
distinction at the beginning of this section that the degree of~$s_1$ is bounded by~$\lambda$ given
in~\cite[Theorem~8]{Abramov2001}. The polynomial~$s_2$ is equal to~$b$ in~\eqref{EQ:ap} whose
degree is minimal by~\cite[Theorem~3]{Abramov2001}. Moreover,~$\sigma_y(G)/G$ is shift-reduced,
because~$\sigma_y(H)/H$ is. These are exactly the same required properties of the output of the
original version~\cite{Abramov2001}.

\end{remark}

    \section{Sum of two residual forms}\label{SECT:sum}
    To compute telescopers for bivariate hypergeometric terms by the modified Abramov-Petkov\v{s}ek reduction,
    we are confronted with the difficulty that the sum of two residual forms is not necessarily
    a residual form. This is because the least common multiple of two shift-free polynomials is not necessarily shift-free.

    The goal of this section is to show that the sum of two residual forms is congruent
    to a residual form modulo~$\bV_K$.

    \begin{example}
        Let $K{=}1/y$,~$r{=}1/(2y+1)$ and~$s{=}1/(2y+3)$. Then both~$r$ and~$s$ are residual forms w.r.t.~$K$,
        but their sum is not, because the denominator~$(2y+1)(2y+3)$ is not shift-free.
        However, we can still find an equivalent residual form. For example, we have $r + s \equiv -1/(2(2y+1)) + 1/2y \mod\bV_K$.
        Note that the residual form is not unique. Another possible choice is $r+s\equiv 1/(3(2y+3)) + 1/3y \bmod\bV_K$.
    \end{example}

    Let~$f$ and~$g$ be two nonzero polynomials in~$\bF[y]$. We say that~$f$ and~$g$ are {\em shift-coprime} if~$\gcd\left(f, \sigma_y^\ell(g)\right) = 1$ for all nonzero integer~$\ell$. Assume that both~$f$ and~$g$ are shift-free. By polynomial factorization
    and dispersion computation, one can uniquely decompose
    \begin{equation} \label{EQ:shiftcoprime}
        g = \tilde{g} \sigma_y^{\ell_1}\left(p_1^{m_1} \right) \cdots \sigma_y^{\ell_k}\left(p_k^{m_k} \right),
    \end{equation}
    where~$\tilde{g}$ is shift-coprime with~$f$, $p_1, \ldots, p_k$ are distinct, monic and irreducible factors of~$f$, $\ell_1, \ldots, \ell_k$
    are nonzero integers, $m_1, \ldots, m_k$ are multiplicities of~$\sigma_y^{\ell_1}\left(p_1 \right)$, $\ldots$,
    $\sigma_y^{\ell_k}\left(p_k \right)$ in~$g$, respectively. We refer to~\eqref{EQ:shiftcoprime} as the \emph{shift-coprime decomposition} of~$g$ w.r.t.~$f$.
    \begin{remark} \label{RE:coprime}
        The factors $\tilde g, \sigma_y^{\ell_1}\left(p_1^{m_1} \right)$,  \ldots, $\sigma_y^{\ell_k}\left(p_k^{m_k}\right)$ in~\eqref{EQ:shiftcoprime} are pairwise coprime, since~$f$
        and~$g$ are shift-free.
    \end{remark}

    To construct a residual form congruent to the sum of two given residual ones, we need three technical lemmas.
    The first one corresponds to the kernel reduction in~\cite{BCCLX2013}.
    \begin{lemma} \label{LM:kernelreduction}
        With Convention~\ref{CON:kernel}, assume that~$p_1, p_2$ are in~$\bF[y]$ and~$m$ in~$\bN$.
        Then there exist~$q_1, q_2$ in~$\bW_K$ s.t.
        \[\frac{p_1}{\prod_{i=0}^m \sigma_y^i(v)}{\equiv} \frac{q_1}{v}~{\rm mod}~\bV_K \,\, \text{and} \,\,
        \frac{p_2}{\prod_{j=1}^m \sigma_y^{-j}(u)}{\equiv} \frac{q_2}{v}~{\rm mod}~\bV_K. \]
    \end{lemma}
    \begin{proof} To prove the first congruence, let~$w_m= \prod_{i=0}^m \sigma_y^i(v)$.

        We proceed by induction on~$m$. If~$m=0$, then the conclusion holds by Lemma~\ref{LM:reducedform}.
        Assume that the lemma holds for~$m-1$.
        Consider the equality
        \[ \frac{p_1}{w_m} =  K \sigma_y\left( \frac{s}{w_{m-1}}\right)
        - \frac{s}{w_{m-1}} + \frac{t}{w_{m-1}}, \]
        where~$s, t \in \bF[y]$ are to be determined. This equality holds if and only if~$\sigma_y(s) u+ (t-s) \sigma_y^{m}(v) = p_1$.
        Since~$u$ and~$\sigma_y^{m}(v)$ are coprime, such~$s$ and~$t$ can be computed by the extended Euclidean algorithm.
        Thus,~$p_1/w_m \equiv t/w_{m-1} \mod \bV_K. $
        Consequently,~$p_1/w_m$ has a required residual form by the induction hypothesis.

        To prove the second congruence, we use
        the identity
        $$\frac{p_2}{\sigma_y^{-1}(u)} =
        K \sigma_y\left( - \frac{p_2}{\sigma_y^{-1}(u)} \right) - \left(- \frac{p_2}{\sigma_y^{-1}(u)}\right)
        + \frac{\sigma_y\left( p_2 \right)}{v}, $$
        which implies that~$p_2/\sigma_y^{-1}(u) \equiv \sigma_y\left( p_2 \right)/v~{\rm mod}~\bV_K.$
        By Lemma~\ref{LM:reducedform}, there exists~$q_2 \in \bW_K$ s.t.~$q_2/v$ is a residual form of~$p_2/\sigma_y^{-1}(u)$
        w.r.t.~$K$.
        Assume that the congruence holds for~$m-1$.  The induction can be completed as in the proof for~$p_1/w_m$.
    \end{proof}
    The next lemma provides us with flexibility to rewrite a rational function modulo~$\bV_K$.
    \begin{lemma} \label{LM:cong}
        Let~$K \in \bF(y)$ be nonzero and shift-reduced.
        Then, for every~$f \in \bF(y)$ and every~$\ell \in \bZ^+$,
        $$ f \equiv \sigma_y^\ell(f) \prod_{i=0}^{\ell-1} \sigma_y^i(K) \equiv
        \sigma_y^{-\ell}(f) \prod_{i=1}^\ell \sigma_y^{-i}\left(\frac{1}{K}\right) \bmod \bV_K.$$
    \end{lemma}
    \begin{proof}
        Let us show the first congruence by induction on~$\ell$. For~$\ell=1$, the identity
        $f = K \sigma_y(-f) - (-f) + \sigma_y(f) K$
        implies that~$f$ is congruent to~$\sigma_y(f) K$ modulo~$\bV_K$.
        Assume that it holds for~$\ell-1$. Set~$w_\ell= \prod_{i=0}^{\ell-1} \sigma_y^i(K)$. Then
        $f$ is congruent to~$\sigma_y^{\ell-1}(f) w_{\ell-1}$ modulo~$\bV_K$ by the induction hypothesis. Moreover,
        $\sigma_y^{\ell-1}(f) w_{\ell-1}$ is congruent to~$\sigma_y^{\ell}(f) w_{\ell}$ by the induction base, in
        which~$f$ is replaced with~$\sigma_y^{\ell-1}(f) w_{\ell-1}$. Hence, $f$ is congruent to $\sigma_y^{\ell}(f) w_{\ell}$
        modulo~$\bV_K$.

        The second congruence can be shown similarly. For $\ell=1$,
        the identity $f = K \sigma_y(r) -r + r$ with $r=\sigma_y^{-1}(f)\sigma_y^{-1}(1/K)$
        implies that $f$ is congruent to~$r$ modulo~$\bV_K$.
        We can then proceed as in the proof of the first congruence.
    \end{proof}
    %
    %

    \begin{lemma} \label{LM:newresidual}
        With Convention~\ref{CON:kernel},
        let~$a, b \in \bF[y]$ with $b{\neq}0$. Assume that~$b$ is shift-free and strongly prime with~$K$.
        Assume further that~$\sigma_y^\ell(b)$ is strongly prime with~$K$ for some integer~$\ell$,
        then~$a/b$ has a residual form~$c/\sigma_y^\ell(b) + q/v$ w.r.t.~$K$,
        where~$c \in \bF[y]$ with~$\deg_y(c) < \deg_y(b)$ and~$q \in \bW_K$.
    \end{lemma}
    \begin{proof} First,
        consider the case in which~$\ell \ge 0$.
        If~$\ell=0$, then there exist~$c, p \in \bF$ with~$\deg_y(c)<\deg_y(b)$ s.t.~$a/b=c/b+p$.
        The lemma follows from Remark~\ref{RE:poly}.

        Assume that~$\ell > 0$. By the first congruence of Lemma~\ref{LM:cong},
        \[ \frac{a}{b} \equiv \sigma_y^\ell\left(\frac{a}{b}\right) \left( \prod_{i=0}^{\ell-1} \sigma_y^i(K) \right)
        = \frac{\sigma_y^\ell(a)}{\sigma_y^\ell(b)}
        \frac{\prod_{i=0}^{\ell-1} \sigma_y^i(u)}{\prod_{i=0}^{\ell-1} \sigma_y^i(v)} \mod \bV_K. \]
        Note that~$\sigma_y^\ell(b)$ is strongly prime with~$v$ by assumption. Then it is coprime
        with the product~$v \sigma_y(v) \cdots \sigma_y^{\ell-1}(v)$. By partial fraction decomposition, we get
        \[ \frac{a}{b} \equiv \frac{\tilde{a}}{\sigma_y^\ell(b)} + \frac{\tilde{q}}{\prod_{i=0}^{\ell-1} \sigma_y^i(v) } \mod \bV_K. \]
        By the first congruence of Lemma~\ref{LM:kernelreduction},
        the second summand in the right-hand side of the above congruence
        can be replaced by a residual form whose denominator is equal to~$v$. The first assertion holds.

        The case in which~$\ell<0$ can be handled in the same way, in which the second congruences of Lemmas~\ref{LM:cong}
        and~\ref{LM:kernelreduction} will be used instead of the first ones in these lemmas.
    \end{proof}

    We are ready to present the main result of this section.
    \begin{theorem}  \label{TH:add}
        With Convection~\ref{CON:kernel}, let~$r$ and~$s$ be two residual forms w.r.t.~$K$.
        Then there exists a residual form~$t$ congruent to~$s$ modulo~$\bV_K$ s.t.,
        for all~$\lambda, \mu \in \bF$,~$\lambda r+ \mu t$ is a residual form w.r.t.~$K$ congruent to~$\lambda r + \mu s$ modulo~$\bV_K$.
    \end{theorem}
    \begin{proof} Let~$r=a/f+p/v$ and~$s=b/g+q/v$, where~$a,f,b,g \in \bF[y]$, $\deg_y(a)<\deg(f)$, $\deg_y(b)<\deg_y(g)$,
        $p, q \in \bW_K$, and~$f, g$ are shift-free and strongly prime with~$K$.

        Assume that~\eqref{EQ:shiftcoprime} is the shift-coprime decomposition of~$g$ w.r.t.~$f$.
        Set~$P_i=\sigma_y^{\ell_i}(p_i)$ for~$i=1$, \ldots,~$k$.
        By Remark~\ref{RE:coprime} and partial fraction decomposition, we have
        \begin{equation} \label{EQ:adddecomp1}
            \frac{b}{g} = \frac{b_0}{\tilde g} + \sum_{i=1}^k \frac{b_i}{P_i^{m_i}},
        \end{equation}
        where~$b_0, b_1, \ldots, b_k \in \bF[y]$. Note that~$p_i=\sigma_y^{-\ell_i}(P_i)$, which
        is a factor of~$f$. Thus it is strongly prime with~$K$. So we can apply Lemma~\ref{LM:newresidual} to each
        fraction~$b_i/P_i^{m_i}$ in~\eqref{EQ:adddecomp1} to get
        \begin{equation} \label{EQ:adddecomp2}
            \frac{b}{g} \equiv  \frac{b_0}{\tilde g} + \sum_{i=1}^k \frac{b_i^\prime}{p_i^{m_i}} + \frac{q^\prime}{v} \mod \bV_K,
        \end{equation}
        where~$b_1^\prime, \ldots, b_k^\prime \in \bF[y]$ and~$q^\prime \in \bW_K$.

        Set~$h=\tilde g \prod_{i=1}^k p_i^{m_i}$. Then~$h$ is shift-free and strongly prime with~$K$ as both~$f$ and~$g$ are.
        Since~$f$ is shift-free, all its factors are shift-coprime with~$f$,
        so are the~$p_i$'s, and so is~$h$. Let~$t$ be the sum of~$q/v$ and the rational function in the right-hand
        side of~\eqref{EQ:adddecomp2}. Then there exist~$b^*$ in~$\bF[y]$ with~$\deg_y(b^*) {<} \deg_y(h)$ and~$q^*$ in~$\bW_K$
        s.t.~$t{=}b^*/h{+}q^*/v$. Since~$f$ and~$h$ are shift-coprime, their least common multiple is shift-free.
        Therefore,~$\lambda r + \mu t$ is a residual form w.r.t.~$K$, and $\lambda r + \mu t$ is congruent
        to~$\lambda r + \mu s \bmod \bV_K$.
    \end{proof}

    \section{Telescoping via reductions} \label{SECT:tele}
    Let~$\cC$ be a field of characteristic zero, and~$\cC(x, y)$
    be the field of rational functions in~$x$ and~$y$ over~$\cC$.
    Let~$\sigma_x, \sigma_y$ be the shift operators w.r.t.~$x$ and~$y$, respectively, defined by,
    \[\sigma_x(f(x, y)) = f(x+1, y) \,\,\, \text{and} \, \,\, \sigma_y(f(x, y)) = f(x, y+1), \]
    for any~$f\in \cC(x, y)$. Then the pair~$(\cC(x, y), \{\sigma_x, \sigma_y\})$
    forms a partial difference field.

    \begin{definition}\label{DEF:bihyper}
        Let~$\bD$ be a partial difference ring extension of~$\cC(x, y)$.
        A nonzero element~$T\in \bD$ is called a \emph{hypergeometric term}  over~$\cC(x, y)$ if there exist~$f, g\in \cC(x, y)$
        s.t.~$\sigma_x(T) = fT$ and~$\sigma_y(T) = g T$. We call~$f, g$ the \emph{shift quotients} of~$T$ w.r.t.~$x$
        and~$y$, respectively.
    \end{definition}

    An irreducible polynomial~$p\in \cC[x, y]$ is said to be~\emph{integer-linear}
    over~$\cC$ if there exist~$f\in \cC[z]$, $m, n\in \bZ$ with~$n\geq 0$ and~$\gcd(m, n)=1$,
    s.t.~$p = f(mx + ny)$. A polynomial in~$\cC[x, y]$ is said to be~\emph{integer-linear} over~$\cC$ if all of its irreducible
    factors are integer-linear. A rational function in~$\cC(x, y)$ is said to be~\emph{integer-linear} over~$\cC$ if its
    denominator and numerator are integer-linear.


    Let~$\bF$ be the field~$\cC(x)$, and~$\bF\langle S_x \rangle $ be the ring of linear recurrence operators
    in~$x$, in which the commutation rule is
    that~$S_x r = \sigma_x(r)S_x$ for all~$r\in \bF$.
    The application of
    an operator~$L = \sum_{i=0}^\rho \ell_iS_x^i$ to a hypergeometric term~$T$ is
    defined as~$L(T) =\sum_{i=0}^\rho \ell_i\sigma_x^i(T)$.

    \begin{definition}\label{DEF:tele}
        Let~$T$ be a hypergeometric term over~$\bF(y)$. A nonzero operator~$L \in \lsx$ is called a \emph{telescoper}
        for~$T$ if there exists a hypergeometric term~$G$ s.t.
        $L(T) = \Delta_y(G).$
        We call~$G$ the \emph{certificate} of~$L$.
    \end{definition}

    For hypergeometric terms, telescopers do not always exist. Abramov presented a criterion
    for determining the existence of telescopers in~\cite[Theorem 10]{Abramov2003}.
    Let~$K{=}u/v$ be a kernel of $\sigma_y(T)/T$ and $S$ the corresponding shell.
    Applying the modified Abramov-Petkov\v{s}ek reduction w.r.t.~$y$ to~$T$
    yields $T = \Delta_y(uH) + rH$, where $u\in \bF(y)$, $H = T/S$, and $r=a/b+q/v$ is the residual form of~$S$ w.r.t.~$K$.
    By Abramov's criterion, $T$ has a telescoper if and only if $b$ is integer-linear over~$\cC$.
    When telescopers exist, Zeilberger's algorithm~\cite{zeilberger90a} constructs a telescoper
    for~$T$ by iteratively using the Gosper algorithm to detect the summability of~$L(T)$ for
    an ansatz~$L=\sum_{i=0}^\rho \ell_i S_x^i\in \lsx$.

    Following the creative telescoping algorithms based on Hermite reductions~\cite{BCCL2010, ChenThesis, bostan13, BCCLX2013}
    in the continuous case, we use the modified Abramov-Petkov\v{s}ek reduction to
    develop a telescoping algorithm, which is outlined below.

    \medskip \noindent
    {\bf {ReductionCT}}: Given a hypergeometric term~$T$ with shift
    quotients~$f=\sigma_x(T)/T$ and~$g=\sigma_y(T)/T$ in~$\bF(y)$,
    compute a telescoper of minimal order for~$T$ and its certificate if telescopers exist.
    \begin{enumerate}
        \item Find a kernel~$K$ and shell~$S$ of~$T$ w.r.t.~$y$ s.t.~$T=SH$ with~$K=\sigma_y(H)/H$.

                  \kern-\medskipamount

        \item Apply the modified Abramov-Petkov\v{s}ek reduction to~$T$ to get
        \begin{equation} \label{EQ:adddecomp0}
            T = \Delta_y(u_0 H)+ r_0 H.
        \end{equation}
        If~$r_0=0$, then return~$(1, u_0H)$.

                  \kern-\medskipamount
        \item If the denominator of~$r_0$ is not integer-linear, return \lq\lq No telescoper exists!\rq\rq.

                  \kern-\medskipamount
        \item Set~$N:=\sigma_x(H)/H$ and~$R := \ell_0 r_0$, where~$\ell_0$ is an indeterminate.

        For~$i=1, 2, \ldots$, do
        \begin{enumerate}
            \item[4.1.] View $\sigma_x(r_{i-1})NH$ as a hypergeometric term with kernel $K$ and
            shell $\sigma_x(r_{i-1})N$. Using shell reduction w.r.t.~$K$ and polynomial reduction
            w.r.t.~$K$, find $u_i'\in\bF$ and a residual form $\tilde r_i$
            w.r.t.~$K$ s.t.~$\sigma_x(r_{i-1})NH=\Delta_y(u_i'H)+\tilde r_iH.$

                  \kern-\smallskipamount
            \item[4.2.] Set $\tilde u_i=\sigma_x(u_{i-1})N+u_i'$, so that
            \begin{equation} \label{EQ:adddecompp}
                \sigma_x^i(T)=\Delta_y(\tilde u_iH)+\tilde r_iH.
            \end{equation}

                  \kern-\smallskipamount
            \item[4.3.] Follow the proof of~Theorem~\ref{TH:add} to compute~$u_i$ and~$r_i$ in~$\bF(y)$
            s.t.~$r_i \equiv \tilde{r}_i \bmod \bV_K,$
            \begin{equation} \label{EQ:adddecomp}
                \sigma_x^i(T) = \Delta_y(u_iH) + r_iH,
            \end{equation}
            and that~$R + \ell_i r_i$ is a residual form w.r.t.~$K$, where~$\ell_i$ is an indeterminate.

                  \kern-\smallskipamount
            \item[4.4.] Update~$R$ to be~$R + \ell_i r_i$.

                  \kern-\smallskipamount
            \item[4.5.] Find~$\ell_j\in \bF$ s.t.~$R= 0$ by solving a linear system in~$\ell_0, \ldots, \ell_i$ over~$\bF$. If there is a nontrivial
            solution, return~$\left(\sum_{j=0}^i \ell_j S_x^j, \, \sum_{j=0}^i \ell_ju_jH \right)$.
        \end{enumerate}
    \end{enumerate}
    \begin{theorem}\label{THM:redct}
        Let~$T$ be a hypergeometric term over~$\bF(y)$. If~$T$ has a telescoper, then
        the algorithm {\bf {ReductionCT}} terminates and returns a telescoper of minimal order
        for~$T$.
    \end{theorem}
    \begin{proof}
        By Theorem~\ref{THM:iapred}, $r_0=0$ implies that~$1$ is a telescoper for~$T$ of minimal order.

        Let~$r_0$ obtained from step~2 be of the form~$a_0/b_0+q_0/v$, where~$a_0, b_0, v \in \bF[y]$,~$\deg_y(a_0) < \deg_y(b_0)$, $b_0$ is strongly prime with~$K$, $q_0 \in \bW_K$, and~$v$ is the denominator of~$K$. By Ore-Sato's theorem~\cite{Ore1930, Sato1990} on hypergeometric
        terms,~$K$ is integer-linear and so is~$v$. It follows that~$b_0$ is integer-linear if and only if~$b_0v$ is.
        By Abramov's criterion,~$T$ has a telescoper if and only if the denominator of~$r_0$ is integer-linear.
        Thus, steps~$2$ and~$3$ are correct.

        It follows from~\eqref{EQ:adddecomp0} and $\sigma_x(r_0H)=\sigma_x(r_0)NH$ that~\eqref{EQ:adddecompp} holds for~$i=1$.
        By Theorem~\ref{TH:add}, there exists a residual form~$r_1$ w.r.t.~$K$  with~$r_1 \equiv \tilde{r}_1 \bmod \bV_K$
        s.t.~$R+\ell_1 r_1$ is again a residual form for all~$\ell_0, \ell_1 \in \bF$. Indeed, the proofs of the lemmas and Theorem~\ref{TH:add}
        enable us to obtain not only~$r_1$ but also a rational function~$g_1$ s.t.~$\tilde{r}_1 = K\sigma_y(g_1) -g_1 + r_1$. Setting~$u_1=\tilde{u}_1+g_1$,
        we see that~\eqref{EQ:adddecomp} holds for~$i=1$.
        By a direct induction on~$i$,~\eqref{EQ:adddecomp} holds in the loop of step~4.

        Assume that~$L = \sum_{i=0}^\rho c_i S_x^i$ is a telescoper of minimal order for~$T$ with~$c_i \in \bF$ and~$c_\rho \neq 0$.
        Then~$L(T)$ is summable. By Theorem~\ref{THM:iapred},~$\sum_{i=0}^\rho c_i r_i$ is equal to zero.
        Thus, the linear homogeneous system (over~$\bF$) obtained by equating~$\sum_{i=0}^\rho \ell_i r_i$ to zero
        has a nontrivial solution, which yields a telescoper of minimal order.
        %
        %
    \end{proof}
    \begin{remark}\label{RE:redCT}
          The algorithm {\bf ReductionCT} separates the computation of minimal
          telescopers from that of certificates. In applications where the
          certificates are irrelevant, we can drop step 4.2 and in step 4.3
          compute $u_i$ and $r_i$ with $r_i\equiv\tilde r_i\mod\bV_K$ and
          $\sigma_x^i(r_{i-1})NH=\Delta_y(u_iH)+r_iH$ and that $R+\ell_i r_i$ is a
          residual form w.r.t.~$K$, where $\ell_i$ is an indeterminate.
          The rational function $u_i$ can be discarded, and we do not need to calculate
          $\sum_{j=0}^i\ell_j u_jH$ in the end.
    \end{remark}
    \begin{remark}\label{RE:simple}
        Instead of applying the modified Abramov-Petkov\v sek reduction to $\sigma_x(r_{i-1})NH$ in step~4.1, it is also
        possible to apply the reduction to~$\sigma_x^i(T)$, but our experiments suggest
        that this variant takes considerably more time.
    \end{remark}
    \begin{example}
        Let~$T = \binom{x}{y}^3$. Set~$f$ and~$g$ to be~$\sigma_x(T)/T$ and~$\sigma_y(T)/T$, respectively.
Since~$g$ is shift-reduced w.r.t.~$y$, its kernel is equal to~$g$ itself, and the corresponding
shell is equal to~$1$. In step~4, we obtain $\sigma_y^{i}(T) \equiv (q_i/v) H \mod \bU_H$,
        where~$i{=}0,1,2$, $v=(y+1)^3$, $H$ has shift quotient~$g$ w.r.t.~$y$,
        \begin{align*}
            q_0 & =  \tfrac12(x{+}1)(x^2{-}x{+}3y(y{-}x{+}1){+}1), \quad q_1 = (x+1)^3, \quad and  \\
            q_2 &= \frac{(x+1)^3}{(x+2)^2} \left( 11x^2-12xy+17x+20+12y+12y^2 \right).
        \end{align*}
        By finding an~$\bF$-linear dependency among~$q_0, q_1, q_2$, we get
        \[L := (x+2)^2 S_x^2 -(7x^2+21x+16)S_x -8(x+1)^2\]
        is a telescoper of minimal order for~$T$.
    \end{example}

    \section{Implementation and timings} \label{SECT:timings}

    We have implemented our algorithms in Maple. In order to get an idea about their
    efficiency, we compared their runtime and memory requirements to the performance
    of known algorithms. All timings are measured in seconds on a Linux computer
        with 388Gb RAM and twelve 2.80GHz Dual core processors.

    For the first comparison, we considered univariate
    hypergeometric terms of the form
    \[
    T=\frac{f(y)}{g_1(y)g_2(y)}\frac{\Gamma(y-\alpha)}{\Gamma(y-\beta)},
    \]
    where $f\in\set Z[y]$ of
    degree~20, $g_i=p_i\sigma_y^\lambda(p_i)\sigma_y^\mu(p_i)$ with $p_i\in\set
    Z[y]$ of degree~10, $\lambda,\mu\in\set N$, and $\alpha,\beta\in\set Z$. For a
    selection of random terms of this type for different choices of $\mu$
    and~$\lambda$, Table~\ref{TAB:generalt} compares the timings of Maple's
    implementation of the classical Abramov-Petkov\v sek reduction~(AP) and our
    modified version~(MAP). We apply the algorithms to $T$ as
    well as to the summable terms $\sigma_y(T)-T$.

    \begin{table}[ht]
        \centering
        \begin{tabular}{ c | r r | r r}
            &\multicolumn{2}{|c|}{$T$}
            &\multicolumn{2}{|c}{$\sigma_y(T)-T$} \\
            $(\lambda, \mu)$&     AP &    MAP &     AP &  MAP   \\ \hline
            $(0, 0)$       &   0.45 &  0.30  &   4.41 &  2.29      \\
            $(5, 5)$       &   5.94 &  1.21  &  10.19 &  2.40      \\
            $(5, 10)$      &  14.69 &  2.20  &  15.67 &  3.30      \\
            $(10, 10)$     &  17.22 &  2.31  &  17.98 &  2.77      \\
            $(10, 20)$     &  57.05 &  6.20  &  38.03 &  3.80      \\
            $(10, 30)$     & 316.51 & 15.73  &  74.55 &  3.73      \\
            $(10, 40)$     & 514.84 & 32.64  & 134.29 &  3.99      \\\hline
        \end{tabular}
        \caption{\small Comparison of the Abramov-Petkov\v sek reduction and the modified version
            for a collection of non-summable terms $T$ and summable terms $\sigma_y(T)-T$.}
        \label{TAB:generalt}
    \end{table}

    For the second comparison, we considered bivariate hypergeometric terms
    of the form
    \[
    T=\frac{f(x,y)}{g_1(x+y)g_2(2x+y)}\frac{\Gamma(2\alpha x+y)}{\Gamma(x+\alpha y)}
    \]
    with $f\in\set Z[x,y]$ of degree~$n$,
    $g_i=p_i\sigma_z^\lambda(p_i)\sigma_z^\mu(p_i)$ with $p_i\in\set Z[z]$ of
    degree~$m$, and $\alpha,\lambda,\mu\in\set N$. For a selection of random terms
    of this type for different choices of $n,m,\alpha,\mu,\lambda$, Table~\ref{TAB:3}
    compares the timings of Maple's implementation of Zeilberger's algorithm~(Z) and
    two variants of the algorithm \textbf{ReductionCT} from Section~\ref{SECT:tele}:
    For the column RCT$_1$ we computed both the telescoper and certificate, and for
    RCT$_2$ we only compute the telescoper. The difference between these two variants
    comes mainly from the time needed to bring the rational function $u$ in the
    certificate $uH$ on a common denominator. When it is acceptable to keep the
    certificate as an unnormalized linear combination of rational functions, the
    timings are virtually the same as for RCT$_2$.

    \begin{table}[!ht]
        \vspace{-\smallskipamount}
        \begin{center}
            \def\c#1{\hbox to4em{\hss\smash{\raisebox{0.25ex}{#1}}\hss}}
            \begin{tabular}{l|r|r|r|c}
                $(m, n, \alpha,\lambda, \mu)$&         Z &  RCT$_1$ & RCT$_2$ & order \\ \hline
                $(1, 0, 1, 5, 5)$            &     17.12 &     5.00 &     1.80 & 4     \\
                $(1, 0, 2, 5, 5)$            &     74.91 &    26.18 &     5.87 & 6     \\
                $(1, 0, 3, 5, 5)$            &    445.41 &    92.74 &    17.34 &  7     \\
                $(1, 8, 3, 5, 5)$            &    649.57 &   120.88 &    23.59 &  7     \\
                $(2, 0, 1, 5, 10)$           &    354.46 &    58.01 &     4.93 & 4     \\
                $(2, 0, 2, 5, 10)$           &    576.31 &   363.25 &    53.15 & 6     \\
                $(2, 0, 3, 5, 10)$           &   2989.18 &  1076.50 &   197.75 &  7     \\
                $(2, 3, 3, 5, 10)$           &   3074.08 &  1119.26 &   223.41 &  7     \\
                $(2, 0, 1, 10, 15)$          &   2148.10 &   245.07 &    11.22 & 4     \\
                $(2, 0, 2, 10, 15)$          &   2036.96 &  1153.38 &   153.21 & 6     \\
                $(2, 0, 3, 10, 15)$          &  11240.90 &  3932.26 &   881.12 & 7     \\
                $(2, 5, 3, 10, 15)$          &  10163.30 &  3954.47 &   990.60 & 7     \\
                $(3, 0, 1, 5, 10)$           &  18946.80 &   407.06 &    43.01 & 6     \\
                $(3, 0, 2, 5, 10)$           &  46681.30 &  2040.21 &   465.88 & 8     \\
                $(3, 0, 3, 5, 10)$           & 172939.00 &  5970.10 &  1949.71 & 9     \\ \hline
            \end{tabular}
        \end{center}
        \caption{\small Comparison of Zeilberger's algorithm to reduction-based telescoping with and without
            construction of a certificate}
        \label{TAB:3}
    \end{table}


\end{document}